\documentclass[11pt]{article}
\usepackage[T1]{fontenc}
\usepackage[utf8]{inputenc}
\usepackage{lmodern}
\usepackage[letterpaper,margin=1in]{geometry}
\usepackage{amsmath,amssymb,amsthm,mathtools}
\usepackage{microtype}
\usepackage{xcolor}
\usepackage{graphicx}
\usepackage{hyperref}
\hypersetup{colorlinks=true,linkcolor=magenta!75!black,citecolor=blue,urlcolor=blue,
  pdftitle={Energy-constrained two-way capacities of the pure-loss channel},
  pdfauthor={Alessandro Falco, Francesco Anna Mele, Ludovico Lami, Vittorio Giovannetti},pdfsubject={Quantum and secret-key capacities under an average photon-number constraint},
  bookmarksnumbered=true}
\newtheorem{theorem}{Theorem}[section]
\newtheorem{lemma}[theorem]{Lemma}
\newtheorem{corollary}[theorem]{Corollary}
\newtheorem{proposition}[theorem]{Proposition}
\theoremstyle{definition}

\newcommand{\cL}{\mathcal L}
\newcommand{\cH}{\mathcal H}
\newcommand{\id}{\operatorname{id}}
\newcommand{\Tr}{\operatorname{Tr}}
\newcommand{\SEP}{\operatorname{SEP}}
\newcommand{\Sym}{\operatorname{Sym}}
\newcommand{\ER}{E_R}
\newcommand{\ERi}{E_R^\infty}
\newcommand{\EE}{\mathbb E}
\newcommand{\CC}{\mathbb C}
\newcommand{\ket}[1]{\lvert #1\rangle}
\newcommand{\bra}[1]{\langle #1\rvert}
\newcommand{\ip}[2]{\langle #1\vert #2\rangle}
\newcommand{\proj}[1]{\ket{#1}\!\bra{#1}}
\newcommand{\norm}[1]{\lVert #1\rVert}

\newcommand{\dd}{\mathrm d}
\newcommand{\Nhat}{\widehat N}
\newcommand{\nhat}{\hat n}

\allowdisplaybreaks[1]
\author{
Alessandro Falco$^{1,*}$
\quad
Francesco Anna Mele$^{2}$
\quad
Ludovico Lami$^{1}$
\quad
Vittorio Giovannetti$^{3}$
}
\makeatletter
\newcommand{\printauthors}{\@author}
\makeatother
\newcommand{\affiliation}[1]{{\small\itshape #1\par}}

\begin{document}
\begin{center}
  {\Large\bfseries Energy-constrained two-way capacities of the pure-loss channel\par}
  \vspace{12pt}
  {\large\printauthors\par}
  \vspace{5pt}
  \affiliation{$^{1}$Scuola Normale Superiore, Piazza dei Cavalieri 7,
  56126 Pisa, Italy}
  \vspace{2pt}
  \affiliation{$^{2}$Institute for Quantum Information and Matter,
  California Institute of Technology,\\
  Pasadena, California 91125, USA}
  \vspace{2pt}
  \affiliation{$^{3}$NEST, Scuola Normale Superiore and Istituto Nanoscienze-CNR,
Piazza dei Cavalieri 7, I-56126 Pisa, Italy}
  \vspace{3pt}
\begingroup
\renewcommand{\thefootnote}{*}
\footnotetext{\href{mailto:alessandro.falco@sns.it}{alessandro.falco@sns.it}}
\endgroup
\vspace{3pt}
\end{center}
\vspace{3pt}
\begin{abstract}
    We determine the two-way quantum and secret-key capacities of the pure-loss channel under an average input photon-number constraint. For transmissivity $\eta$ and photon budget $N$, both capacities equal $g(N)-g((1-\eta)N)$, where $g$ is the entropy of a thermal state. This coincides with the rate achieved by the known reverse coherent information of a two-mode squeezed vacuum. Our result provides the last missing capacity formula for the celebrated pure-loss channel, completing the picture of bosonic quantum communication under pure loss.
\end{abstract}

\section{Main result}
\label{sec:main}

The single-mode pure-loss channel $\cL_\eta$, with transmissivity
$\eta\in[0,1]$, mixes the input with a vacuum environment on a beam
splitter~\cite{HW01,Weed12}. We use its dilation
\begin{equation}
 V_\eta\ket{k}
 =\sum_{j=0}^{k}\sqrt{\binom{k}{j}\eta^{k-j}(1-\eta)^j}
   \ket{k-j}_B\ket{j}_E,
 \qquad
 \cL_\eta(\rho)=\Tr_E V_\eta\rho V_\eta^\dagger.
 \label{eq:loss}
\end{equation}
Thus $j$ is the number of photons lost. The complementary channel is
$\cL_{1-\eta}$, up to the choice of an environment phase.

We consider the standard energy-constrained, two-way-assisted communication
model~\cite{DSW18,MLG25}. Alice and Bob start with local states and may share
public randomness, but no entanglement or secret key. They alternate $n$
uses of $\cL_\eta$ with arbitrary local operations and classical
communication (LOCC)~\cite{CLMOW14}. At each fixed blocklength, a code
uses finitely many local rounds and finite or countable classical message
alphabets; no uniform bound on the number of rounds is imposed as $n$ grows.
All operations act on trace-class states on separable Hilbert spaces;
their adjoints are normal.
The classical channel is authenticated and public.
Eve receives the
channel environments and the public transcript. The only energy constraint
is
\begin{equation}
 \frac1n\sum_{i=1}^{n}\Tr\!\left[\nhat\,\bar\rho_{S_i}\right]\le N,
 \qquad \nhat=a^\dagger a,
 \label{eq:energy-constraint}
\end{equation}
where $\bar\rho_{S_i}$ is the unconditional input state at use $i$:
previous measurement outcomes have already been averaged over. Local
ancillas and memories are unrestricted.

The capacities $Q_2(\cL_\eta,N)$ and $K(\cL_\eta,N)$ are the supremal
asymptotic rates of entanglement and secret-key generation, respectively,
with error tending to zero. Rates are measured in ebits and secret bits per
channel use. With free classical communication, entanglement generation and
quantum transmission have the same capacity by
teleportation~\cite{BBC93,DSW18}. For secret-key generation, we use the standard correctness and secrecy
error~\cite{PR22}. For entanglement generation, the error is the
infidelity with a maximally entangled state. Both errors are evaluated
on the unconditional output.

\begin{theorem}[Energy-constrained two-way capacities]
\label{thm:main}
For every $\eta\in[0,1]$ and every finite $N\ge0$,
\begin{equation}
 \boxed{\quad
 Q_2(\cL_\eta,N)=K(\cL_\eta,N)
 =g(N)-g\bigl((1-\eta)N\bigr),
 \quad}
 \label{eq:main}
\end{equation}
where $g(x)=(x+1)\log_2(x+1)-x\log_2x$ and $0\log_2 0=0$.
The optimization includes all adaptive LOCC protocols satisfying
\eqref{eq:energy-constraint}.
\end{theorem}

Throughout, logarithms are in base $2$, and all Hilbert spaces are
separable. We write $S(\rho)=-\Tr\rho\log_2\rho$ for von Neumann
entropy and $H(p)=-\sum_a p_a\log_2p_a$ for Shannon entropy. Identity
maps on untouched systems are implicit. We abbreviate
\begin{equation}
 q:=1-\eta,
 \qquad f_\eta(x):=g(x)-g(qx).
 \label{eq:f}
\end{equation}
The formula gives $0$ at $N=0$ or $\eta=0$, and $g(N)$ at $\eta=1$.
For $0\le\eta<1$, it also gives
$\lim_{N\to\infty}f_\eta(N)=-\log_2(1-\eta)$, recovering the
unconstrained capacity~\cite{PLOB17}. Earlier upper bounds for two-way
bosonic communication were obtained from squashed entanglement;
see~\cite{TGW14,GEW16,DSW18}. The achievable function $f_\eta$ comes from
reverse reconciliation and reverse coherent
information~\cite{PGBL09,GPLS09}.

Figure~\ref{fig:capacities-vs-eta} compares these rates with the other
capacities of the pure-loss channel at $N=1$. For $0<\eta\le1/2$, the
unassisted quantum and private classical capacities vanish~\cite{WHG12,WQ18},
whereas $Q_2$ and $K$ remain positive. At $\eta=1/2$, for example,
$Q_2=K\simeq0.6226$.

\begin{figure}[t]
  \centering
  \includegraphics[width=0.94\textwidth]{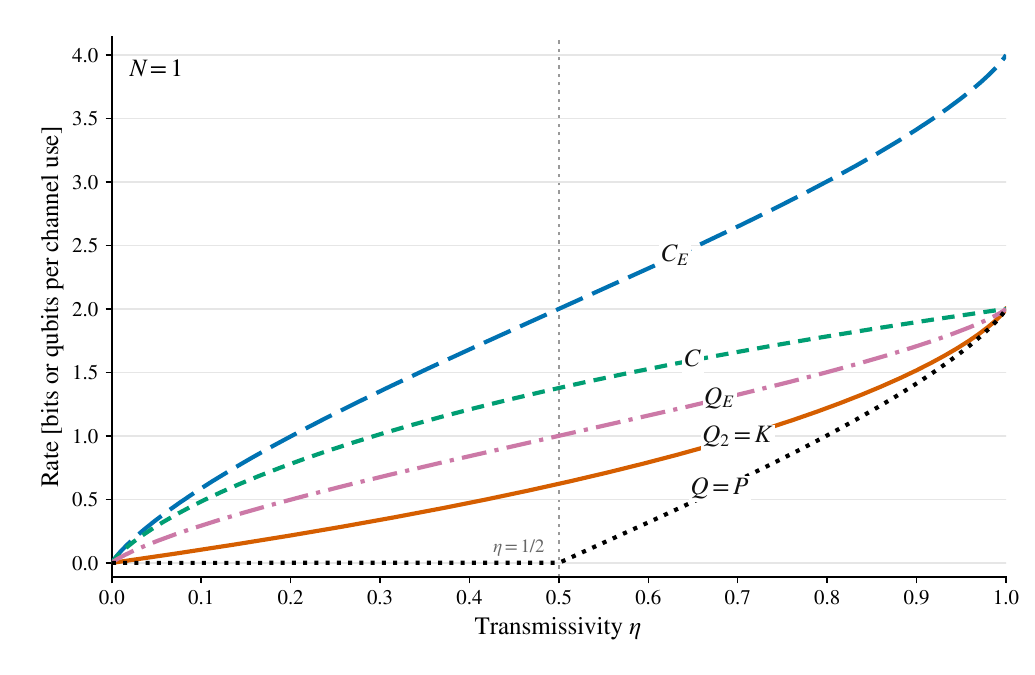}
  \caption{Capacities of the single-mode pure-loss channel as functions of
  transmissivity $\eta$, for an average input photon budget $N=1$.
  All rates are asymptotic, with error tending to zero.
  The two-way quantum and secret-key capacities are
  $Q_2=K=g(N)-g((1-\eta)N)$ (solid), by Theorem~\ref{thm:main}.
  The unassisted classical capacity is $C=g(\eta N)$
  (short-dashed)~\cite{GGLLSY04}.
  With unlimited preshared entanglement, the classical capacity is
  $C_E=g(N)+g(\eta N)-g((1-\eta)N)$
  (long-dashed)~\cite{GLMS03}, and the quantum capacity is
  $Q_E=C_E/2$ (dash-dotted)~\cite{BSST02}.
  The unassisted quantum and private classical capacities are
  $Q=P=\max\{0,g(\eta N)-g((1-\eta)N)\}$
  (dotted)~\cite{WHG12,WQ18}.
  The vertical guide marks $\eta=1/2$; $Q=P=0$ at and below this value.
  Classical and key rates are in bits per channel use; quantum rates are
  in qubits per channel use.}
  \label{fig:capacities-vs-eta}
\end{figure}

\paragraph{Proof outline.}
The lower bound is reverse hashing~\cite{DW05,GPLS09}; we recall it in
Section~\ref{sec:achievability}. For finite input energy and finite input
relative entropy of entanglement, the central converse statement is
\begin{equation}
 \ERi(X:BY)_{(\id\otimes\cL_\eta)(\rho)}
 \le \ERi(XS:Y)_\rho+f_\eta\!\left(\Tr[\nhat\rho_S]\right).
 \label{eq:outline-amortized}
\end{equation}
Here $XS$ belongs to Alice before transmission, $Y$ belongs to Bob, and
$\ERi$ is the regularized relative entropy of entanglement. The inequality
allows arbitrary correlations with both memories. It can therefore be
summed over the uses of an adaptive protocol. We prove it by grouping
independent copies, resolving their total photon number, and bounding the
relative entropy of entanglement of the sector Choi states. The remaining dephasing
cost grows only logarithmically with the number of copies. Finally, a
privacy test converts the entanglement bound into a bound on secret-key
length that is independent of the size of the shield systems.

\section{Relative entropy tools}
\label{sec:tools}

For states $\rho,\sigma$, their quantum relative entropy is
$D(\rho\Vert\sigma)=\Tr\rho(\log_2\rho-\log_2\sigma)$, with its usual
extended-value interpretation. We use its data-processing inequality,
$D(\rho\Vert\sigma)\ge D(\mathcal T(\rho)\Vert\mathcal T(\sigma))$
for every quantum channel $\mathcal T$~\cite{Lin75}, with the
separable-Hilbert-space formulation covered by~\cite{MR17}.

For a bipartite state, the relative entropy of entanglement and its
regularization are~\cite{VPRK97,VP98}
\begin{equation}
 \ER(A:B)_\rho:=\inf_{\sigma\in\SEP(A:B)}D(\rho\Vert\sigma),
 \qquad
 \ERi(A:B)_\rho:=\lim_{r\to\infty}\frac1r\ER(A^r:B^r)_{\rho^{\otimes r}}.
 \label{eq:ER-def}
\end{equation}
The separable states form the trace-norm closed convex hull of product
states. The second limit exists whenever $\ER(\rho)<\infty$, by
subadditivity. We omit the bipartition when it is clear.

We use three standard properties of $\ER$: convexity, subadditivity on
tensor products, and monotonicity on average under local instruments
and LOCC~\cite{VP98}. For completeness, the last property follows by
applying data processing to an instrument with its outcome retained.
For any separable comparison state, the resulting relative entropy is a
classical relative entropy plus an average of conditional relative
entropies. Dropping the first term and optimizing gives
$\sum_a p_a\ER(\rho_a)\le\ER(\rho)$.
Convexity follows from joint convexity of relative entropy, and
subadditivity follows by taking product comparison states. These arguments
also apply on separable infinite-dimensional Hilbert spaces. We will use
that $\ER$ is lower semicontinuous in trace norm~\cite[Corollary~20]{LS23}.
The regularized quantity is monotone under LOCC because this holds for
every tensor power before taking the limit.

\subsection{The cost of local dephasing}

The following finite-dimensional observation is the form of
non-lockability that we need; see also~\cite{HHHO05a}.

\begin{lemma}[Local dephasing]
\label{lem:pinch}
Let $\Delta(\rho)=\sum_a P_a\rho P_a$ be a projective dephasing on one
party of a finite-dimensional bipartite system, and let
$p_a=\Tr(P_a\rho)$. Then
\begin{equation}
 \ER(\rho)\le\ER(\Delta\rho)+D(\rho\Vert\Delta\rho),
 \qquad D(\rho\Vert\Delta\rho)\le H(p).
 \label{eq:pinch}
\end{equation}
\end{lemma}
\begin{proof}
In minimizing $D(\Delta\rho\Vert\sigma)$ over separable $\sigma$, we may
dephase $\sigma$: separability is preserved and data processing cannot
increase the objective. For $\Delta\sigma=\sigma$, block diagonality gives
\[
 D(\rho\Vert\sigma)
 =D(\rho\Vert\Delta\rho)+D(\Delta\rho\Vert\sigma).
\]
Optimizing proves the first inequality. For the second, purify $\rho$ to
$\ket{\psi}$. The nonzero vectors $P_a\ket{\psi}/\sqrt{p_a}$ are orthonormal,
so the relative entropy of $\proj{\psi}$ from its dephased state is
$-\sum_a p_a\log_2p_a=H(p)$. Tracing out the purification and applying data
processing proves the claim.
\end{proof}

\subsection{Simulation from an irreducible symmetry}

Write $\ket{\Phi_d}=d^{-1/2}\sum_{a=1}^d\ket{a}\ket{a}$ and
$\Phi_d=\proj{\Phi_d}$. A finite-dimensional channel $\mathcal T$ is
simulable by LOCC using its normalized Choi state
$\Omega_{\mathcal T}=(\id\otimes\mathcal T)(\Phi_d)$ whenever it is
covariant under a compact group whose input representation is irreducible:
\begin{equation}
 \mathcal T(U_g\rho U_g^\dagger)=W_g\mathcal T(\rho)W_g^\dagger.
 \label{eq:covariance}
\end{equation}
The required simulation is given in~\cite[Proposition~24]{GBW21}.
To check its normalization,
Alice measures the input and her half of $\Omega_{\mathcal T}$ with the
Bell POVM, with normalized Haar measure $\mu$,
\[
 d^2(I\otimes\bar U_g)\Phi_d(I\otimes\bar U_g)^\dagger\,\dd\mu(g).
\]
Schur's lemma makes its integral the identity. Finite-dimensional
convexity also allows a finite weighted sum with the same operator
average, so finitely many measurement outcomes suffice. Conditional on $g$,
Bob's
unnormalized state is $\mathcal T(U_g\rho U_g^\dagger)\dd\mu(g)$; he applies
$W_g^\dagger$. This implements $\mathcal T$ exactly, also on inputs
correlated with other systems. Consequently, for any such input $\rho$,
\begin{equation}
 \ER(\text{output})\le\ER(\text{input})+\ER(\Omega_{\mathcal T}).
 \label{eq:simulation-cost}
\end{equation}
Indeed, the simulation is LOCC, and its initial state is the product of
$\rho$ with $\Omega_{\mathcal T}$. The resource term in
\eqref{eq:simulation-cost} is independent of the local memories.

\section{A fixed number of photons}
\label{sec:sectors}

Fix $m\ge1$ modes. The sector with $k$ photons is
\begin{equation}
 \cH_{m,k}=\Sym^k(\CC^m),
 \qquad d_k:=\dim\cH_{m,k}=\binom{m+k-1}{k},
 \qquad \pi_k:=\frac{\Pi_{m,k}}{d_k}.
 \label{eq:sectors}
\end{equation}
Here $\Pi_{m,k}$ is the projector onto that sector; $m$ is suppressed in
$d_k$ and $\pi_k$. The binomial coefficient counts the ways to distribute
$k$ photons among $m$ occupation numbers. Passive mode transformations $U\in U(m)$ act as
$U^{\otimes k}$ restricted to $\Sym^k(\CC^m)$, an irreducible
representation~\cite{Har13}.

For a unit vector $u\in\CC^m$, let
\begin{equation}
 \ket{u;k}:=\frac{\bigl(\sum_{a=1}^m u_a a_a^\dagger\bigr)^k}{\sqrt{k!}}
 \ket{0}.
 \label{eq:condensate}
\end{equation}
This is the normalized state of $k$ photons in the same mode $u$.
Averaging uniformly over $u$ gives
\begin{equation}
 \int\proj{u;k}\,\dd\mu(u)=\pi_k.
 \label{eq:resolution}
\end{equation}
The integral is invariant under every mode transformation, so
irreducibility makes it proportional to $\Pi_{m,k}$; its trace fixes the
factor $1/d_k$. This is the standard symmetric-subspace
resolution~\cite{Har13}.

\subsection{Conditioning on the number lost}

With $k$ input photons, the total number $J$ lost has distribution
\begin{equation}
 b_{k,j}:=\Pr(J=j\mid k)=\binom{k}{j}q^j\eta^{k-j},
 \qquad 0\le j\le k,
 \label{eq:binomial}
\end{equation}
independently of the state within the sector. Here and below endpoint
probabilities are understood by continuity. To verify this for arbitrary
inputs, write a Fock occupation vector as $\mathbf a=(a_1,\ldots,a_m)$ with
$|\mathbf a|=k$. The normalized isometry for the branch with $j$ losses is
\begin{equation}
 V_{k,j}\ket{\mathbf a}
 =\sum_{\substack{\mathbf t\le\mathbf a\\|\mathbf t|=j}}
 \sqrt{\frac{\prod_{s=1}^m\binom{a_s}{t_s}}{\binom{k}{j}}}
 \ket{\mathbf a-\mathbf t}_B\ket{\mathbf t}_E.
 \label{eq:conditional-isometry}
\end{equation}
Different $\mathbf a$ give orthogonal outputs. Their squared norms are one
by Vandermonde's identity,
$\sum_{|\mathbf t|=j}\prod_s\binom{a_s}{t_s}=\binom{k}{j}$.
Expanding \eqref{eq:loss} on all $m$ modes therefore gives the branch
amplitude $\sqrt{b_{k,j}}V_{k,j}$, proving~\eqref{eq:binomial}.
Expanding~\eqref{eq:condensate} also gives the useful identity
\begin{equation}
 V_{k,j}\ket{u;k}=\ket{u;k-j}_B\ket{u;j}_E.
 \label{eq:splitting}
\end{equation}
Together with~\eqref{eq:resolution}, this implies that the environment
marginal of $V_{k,j}\pi_kV_{k,j}^\dagger$ is $\pi_j$.

\subsection{An explicit separable comparison state}

Let $\mathcal T_{m,k}$ be $\cL_\eta^{\otimes m}$ restricted to $\cH_{m,k}$,
and let $\Omega_{m,k}$ be its normalized Choi state. Its output lies in the
finite-dimensional space $\bigoplus_{\ell=0}^k\cH_{m,\ell}$.
The channel is covariant under passive mode transformations, with an
irreducible input representation. Section~\ref{sec:tools} therefore gives
an exact simulation using $\Omega_{m,k}$.

\begin{lemma}[Relative entropy of the sector Choi state]
\label{lem:sector-cost}
For all $m\ge1$ and $k\ge0$,
\begin{equation}
 \ER(\Omega_{m,k})\le c_m(k),
 \qquad
 c_m(k):=\log_2d_k-\sum_{j=0}^k b_{k,j}\log_2d_j.
 \label{eq:sector-cost}
\end{equation}
\end{lemma}
\begin{proof}
Fix $j$ and put $\ell=k-j$. Let
\[
 \ket{\Psi_{k,j}}=(I_R\otimes V_{k,j})\ket{\Phi_{d_k}},
 \qquad \Omega_{k,j}=\Tr_E\proj{\Psi_{k,j}}.
\]
The environment marginal is $\pi_j$, by~\eqref{eq:splitting}. Choose an
orthonormal basis $\{\ket{t}\}_{t=1}^{d_j}$ of $\cH_{m,j}$ and set
\begin{equation}
 C_t=(I_B\otimes\bra{t}_E)V_{k,j},
 \qquad
 \ket{\chi_t}=\sqrt{d_j}(I_R\otimes C_t)\ket{\Phi_{d_k}}.
 \label{eq:chi}
\end{equation}
Since the environment marginal is $\pi_j$, these vectors are orthonormal.
Thus, with $P_{k,j}:=\sum_t\proj{\chi_t}$,
\begin{equation}
 \Omega_{k,j}=\frac{P_{k,j}}{d_j}.
 \label{eq:flat-choi}
\end{equation}
Consider the separable state
\begin{equation}
 \sigma_{k,j}:=\int
 \proj{\bar u;k}_R\otimes\proj{u;\ell}_B\,\dd\mu(u),
 \label{eq:separable-comparison}
\end{equation}
where conjugation is in the Fock basis. We claim that
\begin{equation}
 \sigma_{k,j}P_{k,j}=\frac1{d_k}P_{k,j}.
 \label{eq:key-identity}
\end{equation}
Indeed, \eqref{eq:splitting} gives
$C_t\ket{u;k}=\ip{t}{u;j}\ket{u;\ell}$ and hence
\[
 (\bra{\bar u;k}\otimes\bra{u;\ell})\ket{\chi_t}
 =\sqrt{\frac{d_j}{d_k}}\,\ip{t}{u;j}.
\]
Using this in~\eqref{eq:separable-comparison},
\begin{align}
 \sigma_{k,j}\ket{\chi_t}
 &=\sqrt{\frac{d_j}{d_k}}
   \int\ket{\bar u;k}\otimes C_t\ket{u;k}\,\dd\mu(u)\nonumber\\
 &=\frac{\sqrt{d_j}}{d_k}(I\otimes C_t)\ket{\Phi_{d_k}}
 =\frac1{d_k}\ket{\chi_t}.
 \label{eq:identity-proof}
\end{align}
For the second equality, vectorizing~\eqref{eq:resolution} gives
$\int\ket{\bar u;k}\otimes\ket{u;k}\,\dd\mu(u)
 =d_k^{-1/2}\ket{\Phi_{d_k}}$.

Equations~\eqref{eq:flat-choi} and~\eqref{eq:key-identity} now imply
\begin{equation}
 \ER(\Omega_{k,j})
 \le D(\Omega_{k,j}\Vert\sigma_{k,j})
 =\log_2d_k-\log_2d_j.
 \label{eq:conditional-cost}
\end{equation}
Finally, $\Omega_{m,k}=\sum_j b_{k,j}\Omega_{k,j}$, with the different
output sectors embedded in the same space. Convexity of $\ER$ and
\eqref{eq:conditional-cost} prove~\eqref{eq:sector-cost}.
\end{proof}

\section{Entanglement increase under an energy constraint}
\label{sec:energy}

We now allow superpositions of different photon numbers and correlations
with either party's memory.

\begin{lemma}[Block bound]
\label{lem:block}
Let $\rho_{XS^mY}$ be a state on separable Hilbert spaces with
$\ER(XS^m:Y)_\rho<\infty$ and
$\Tr[\Nhat_{S^m}\rho]=m\nu<\infty$, where
$\Nhat_{S^m}=\sum_{a=1}^m\nhat_{S_a}$. Set
$\omega=(\id_{XY}\otimes\cL_\eta^{\otimes m})(\rho)$. Then
\begin{equation}
 \ER(X:B^mY)_\omega
 \le \ER(XS^m:Y)_\rho+m f_\eta(\nu)+g(qm\nu).
 \label{eq:block-bound}
\end{equation}
\end{lemma}
\begin{proof}
We first suppose that $X,Y$ are finite-dimensional and that $\rho$ has
finite photon-number support. All subsequent entropies are then ordinary
finite quantities. The infinite-dimensional extension is given at the
end of the proof.

\paragraph{Step 1: pay for coherences between sectors.}
Let $\Delta_S$ and $\Delta_B$ dephase the total photon number of the
$m$ input and output modes, respectively. Phase covariance gives
\begin{equation}
 \Delta_B\circ\cL_\eta^{\otimes m}
 =\cL_\eta^{\otimes m}\circ\Delta_S.
 \label{eq:phase-covariance}
\end{equation}
For example, both dephasings are averages over
$e^{i\theta\Nhat}$, and the channel intertwines these phase rotations.
Let $K$ be the input total photon number, with probabilities $p_k$,
and write $\Delta_S\rho=\sum_k p_k\rho_k$ for the normalized conditional
states. Put $\omega_k=(\id\otimes\cL_\eta^{\otimes m})(\rho_k)$.
Lemma~\ref{lem:pinch}, data processing, and
\eqref{eq:phase-covariance} yield
\begin{equation}
 \ER(\omega)\le\ER(\Delta_B\omega)+H(K).
 \label{eq:pay-coherence}
\end{equation}
Explicitly, the additional term is bounded as
$D(\omega\Vert\Delta_B\omega)
 \le D(\rho\Vert\Delta_S\rho)\le H(K)$.
By convexity, the sector simulation, and average LOCC monotonicity,
\begin{align}
 \ER(\Delta_B\omega)
 &\le\sum_k p_k\ER(\omega_k)\nonumber\\
 &\le\sum_k p_k\bigl(\ER(\rho_k)+c_m(k)\bigr)
 \le\ER(\rho)+\EE c_m(K).
 \label{eq:sector-average}
\end{align}
Here the bipartition is $XS^m:Y$ before transmission and $X:B^mY$ after
transmission. In particular, measuring $K$ is a local instrument on
Alice's side. Combining the two bounds gives
\begin{equation}
 \ER(\omega)-\ER(\rho)
 \le H(K)+\EE\log_2d_K-\EE\log_2d_J,
 \label{eq:photon-bound}
\end{equation}
where the joint law is $\Pr(K=k,J=j)=p_kb_{k,j}$.

\paragraph{Step 2: express the photon bound as an entropy difference.}
Introduce the auxiliary states
\begin{equation}
 \zeta:=\sum_k p_k\pi_k,
 \qquad
 \zeta':=\cL_q^{\otimes m}(\zeta)=\sum_j\Pr(J=j)\pi_j.
 \label{eq:aux-states}
\end{equation}
The output identity follows from~\eqref{eq:resolution} and
\eqref{eq:splitting}. Orthogonality of the sectors gives
\begin{equation}
 S(\zeta)=H(K)+\EE\log_2d_K,
 \qquad S(\zeta')=H(J)+\EE\log_2d_J.
 \label{eq:sector-entropies}
\end{equation}
Therefore the right-hand side of~\eqref{eq:photon-bound} is
$S(\zeta)-S(\zeta')+H(J)$.

\paragraph{Step 3: compare with thermal states.}
The one-mode thermal state with mean photon number $x$ is
\begin{equation}
 \tau_x=\frac1{x+1}\sum_{k\ge0}\left(\frac{x}{x+1}\right)^k\proj{k},
 \qquad S(\tau_x)=g(x).
 \label{eq:thermal}
\end{equation}
For $x>0$ and any $m$-mode state $\theta$ of total mean photon number
$mx$, the logarithm of the thermal reference gives
\begin{equation}
 D(\theta\Vert\tau_x^{\otimes m})=m g(x)-S(\theta).
 \label{eq:thermal-relative}
\end{equation}
In particular, positivity of relative entropy proves
$S(\theta)\le m g(x)$. Binomial thinning of the geometric distribution
in~\eqref{eq:thermal} gives $\cL_q(\tau_x)=\tau_{qx}$.
The states in~\eqref{eq:aux-states} have total means $m\nu$ and $qm\nu$,
respectively. For $q>0$ and $\nu>0$, data processing gives
\begin{align}
 m g(\nu)-S(\zeta)
 &=D(\zeta\Vert\tau_\nu^{\otimes m})\nonumber\\
 &\ge D(\zeta'\Vert\tau_{q\nu}^{\otimes m})
 =m g(q\nu)-S(\zeta').
 \label{eq:thermal-dpi}
\end{align}
It follows that $S(\zeta)-S(\zeta')\le m f_\eta(\nu)$.
If $q\nu=0$, the state $\zeta'$ is vacuum, and the same inequality follows
from the entropy bound just proved; if $\nu=0$, both states are vacuum.
Finally, $J$ is a nonnegative integer with mean $qm\nu$. Comparing its
distribution with the geometric distribution of that mean gives
$H(J)\le g(qm\nu)$, again by nonnegativity of relative entropy.
Substitution in~\eqref{eq:photon-bound} proves~\eqref{eq:block-bound}.

\paragraph{Step 4: remove all cutoffs.}
Choose increasing finite-rank projections $P_L^X\to I_X$ and
$P_L^Y\to I_Y$ strongly, and define the product projection
\[
 T_L=P_L^X\otimes\Pi_{m,\le L}^{S^m}\otimes P_L^Y,
 \qquad p_L=\Tr(T_L\rho),
 \qquad \rho_L=T_L\rho T_L/p_L.
\]
For sufficiently large $L$, $p_L>0$. We have $p_L\to1$ and
$\rho_L\to\rho$ in trace norm. The projections commute with the photon
number, so monotone convergence gives
$\nu_L:=m^{-1}\Tr(\Nhat\rho_L)\to\nu$.
The successful projection is a local filtering outcome; hence
$p_L\ER(\rho_L)\le\ER(\rho)$.
Writing $\omega_L=(\id\otimes\cL_\eta^{\otimes m})(\rho_L)$ and applying
the finite-dimensional bound,
\[
 \ER(\omega_L)\le\frac{\ER(\rho)}{p_L}
          +m f_\eta(\nu_L)+g(qm\nu_L).
\]
Now $\omega_L\to\omega$ in trace norm. By lower semicontinuity of
$\ER$~\cite[Corollary~20]{LS23},
\[
 \ER(\omega)\le\liminf_{L\to\infty}\ER(\omega_L)
 \le\ER(\rho)+m f_\eta(\nu)+g(qm\nu),
\]
where the last inequality uses $p_L\to1$, $\nu_L\to\nu$, and continuity
of $f_\eta,g$. No entropy assumption on $X$ or $Y$ is needed.
\end{proof}

\begin{corollary}[Amortized regularized bound]
\label{cor:amortized}
Let $\rho_{XSY}$ satisfy $\ER(XS:Y)_\rho<\infty$ and
$\Tr(\nhat\rho_S)=\nu<\infty$. For
$\omega=(\id_{XY}\otimes\cL_\eta)(\rho)$,
\begin{equation}
 \ERi(X:BY)_\omega\le\ERi(XS:Y)_\rho+f_\eta(\nu).
 \label{eq:amortized}
\end{equation}
\end{corollary}
\begin{proof}
Apply Lemma~\ref{lem:block} to $r$ copies of $\rho$, grouping all of
Alice's memories and all of Bob's memories. Their $r$ input modes have
total mean photon number $r\nu$, so
\[
 \ER(\omega^{\otimes r})\le\ER(\rho^{\otimes r})
             +r f_\eta(\nu)+g(qr\nu).
\]
Divide by $r$ and take $r\to\infty$. The last term divided by $r$
tends to zero, since $g(x)=O(\log(x+1))$ as $x\to\infty$.
The single-copy block bound also ensures $\ER(\omega)<\infty$, so both
regularizations exist. This proves~\eqref{eq:amortized}.
\end{proof}

\section{Private states and the capacity theorem}
\label{sec:privacy}

We next relate $\ERi$ to an approximate secret key. This step keeps the
local shield systems and places no bound on their dimensions.

\subsection{A privacy test that survives regularization}

A private state with key dimension $d$ has the form~\cite{HHHO05b,HHHO09,WTB17}
\begin{equation}
 \gamma_{AA'BB'}=U(\Phi_{d,AB}\otimes\theta_{A'B'})U^\dagger,
 \qquad
 U=\sum_{a,b=1}^d\proj{a,b}_{AB}\otimes U_{ab},
 \label{eq:private-state}
\end{equation}
where the $U_{ab}$ are unitaries and $A',B'$ are the shields.
Measuring $A,B$ in the displayed basis gives a uniform common key,
independent of every purifying system.
Its privacy test is the projector
\begin{equation}
 \Pi_\gamma=U(\Phi_{d,AB}\otimes I_{A'B'})U^\dagger.
 \label{eq:privacy-test}
\end{equation}
We use squared fidelity,
$F(\rho,\sigma)=\norm{\sqrt\rho\sqrt\sigma}_1^2$.
The standard privacy-test properties are~\cite[Lemmas~9--10]{WTB17}
\begin{equation}
 \Tr(\Pi_\gamma\omega)\ge F(\omega,\gamma),
 \qquad
 \Tr(\Pi_\gamma\sigma)\le\frac1d
 \quad\text{for every }\sigma\in\SEP(AA':BB').
 \label{eq:privacy-properties}
\end{equation}
The first inequality follows by applying the test to both states and
using fidelity monotonicity. To check the second, take a normalized
product vector $\ket{\alpha}_{AA'}\ket{\beta}_{BB'}$ and write
$\ket{\alpha}=\sum_a\ket a\ket{\alpha_a}$ and
$\ket{\beta}=\sum_a\ket a\ket{\beta_a}$. Then
\begin{equation}
 \bra{\alpha,\beta}\Pi_\gamma\ket{\alpha,\beta}
 =\frac1d\norm{\sum_a U_{aa}^\dagger
                 (\ket{\alpha_a}\otimes\ket{\beta_a})}^2
 \le\frac1d\left(\sum_a\norm{\alpha_a}\norm{\beta_a}\right)^2
 \le\frac1d.
 \label{eq:privacy-soundness-proof}
\end{equation}
The triangle and Cauchy--Schwarz inequalities give the two bounds.
Convexity and trace-norm closure extend them to every separable state,
including infinite-dimensional shields.

\begin{lemma}[Regularized privacy bound]
\label{lem:privacy}
Suppose $d\ge2$, $0\le\varepsilon\le\tfrac12$, and
$F(\omega,\gamma)\ge1-\varepsilon$ for a private state of key dimension
$d$. If $\ER(\omega)<\infty$, then
\begin{equation}
 \ERi(\omega)\ge(1-\varepsilon)\log_2d-h_2(\varepsilon),
 \label{eq:privacy-bound}
\end{equation}
where $h_2(x)=-x\log_2x-(1-x)\log_2(1-x)$.
\end{lemma}
\begin{proof}
Put $\Pi=\Pi_\gamma$ and $p=\Tr(\Pi\omega)\ge1-\varepsilon$.
For any state $\sigma_r$ separable between all of Alice's systems and all
of Bob's systems, and any subset $S\subseteq\{1,\ldots,r\}$,
\begin{equation}
 \Tr\!\left[\sigma_r\prod_{a\in S}\Pi_a\right]\le d^{-|S|}.
 \label{eq:subset-tests}
\end{equation}
Indeed, the marginal on $S$ is separable, and the product test is itself
a privacy test with key dimension $d^{|S|}$. For $t\ge0$, expansion with
nonnegative coefficients therefore gives
\begin{equation}
 \Tr\!\left[\sigma_r\prod_{a=1}^r(I+(2^t-1)\Pi_a)\right]
 \le\sum_{S\subseteq\{1,\ldots,r\}}(2^t-1)^{|S|}d^{-|S|}
 =\left(1+\frac{2^t-1}{d}\right)^r.
 \label{eq:mgf}
\end{equation}
Measure the $r$ tests separately, and let $Z$ count their successes.
Let $P,Q$ be the resulting distributions for $\omega^{\otimes r}$ and
$\sigma_r$. Then $\EE_P Z=rp$, and~\eqref{eq:mgf} bounds
$\EE_Q2^{tZ}$. Data processing and the classical variational inequality
$D(P\Vert Q)\ge t\EE_P Z-\log_2\EE_Q2^{tZ}$ yield
\begin{equation}
 D(\omega^{\otimes r}\Vert\sigma_r)
 \ge r\left[tp-\log_2\!\left(1+\frac{2^t-1}{d}\right)\right].
 \label{eq:variational-test}
\end{equation}
The classical inequality follows by comparing $P$ with the probability
distribution proportional to $2^{tZ}Q$ and using nonnegativity of relative
entropy. Since $p\ge1-\varepsilon\ge1/d$, the supremum of the bracket
over $t\ge0$ is the binary relative entropy $d_2(p\Vert1/d)$.
Moreover,
\[
 d_2(p\Vert1/d)
 =p\log_2d+(1-p)\log_2\!\frac{d}{d-1}-h_2(p)
 \ge(1-\varepsilon)\log_2d-h_2(\varepsilon).
\]
Taking the infimum over $\sigma_r$ in~\eqref{eq:variational-test}, dividing
by $r$, and then taking $r\to\infty$ proves~\eqref{eq:privacy-bound}.
\end{proof}

\subsection{Converse for adaptive protocols}
\label{sec:converse}

To apply the privacy bound to an actual key protocol, keep the systems
that each local operation would otherwise discard, and postpone the final
key measurements. Public messages are still sent to Eve, and local copies
of their values are retained. Tracing out these extra local systems recovers
the original protocol after every step: all transmitted states, and hence
all their mean energies, are unchanged. This is the usual private-state
representation~\cite{HHHO09}; its energy-preserving implementation on
separable Hilbert spaces is described in~\cite[Sections~II.F and~IV]{DSW18}.

Here is the error conversion. Let $\ket\Psi_{ABA'B'E_{\rm all}}$ be the
resulting pure state, with $A,B$ the unmeasured key registers, and let
$\omega=\Tr_{E_{\rm all}}\proj\Psi$. Write
$\ket\Psi=\sum_{a,b}\ket{a,b}\ket{\psi_{ab}}$, where the vectors
$\ket{\psi_{ab}}$ need not be normalized. The post-measurement state is
$\kappa$. The Fuchs--van de Graaf inequality~\cite{FvdG99} gives root
fidelity at least $1-e$ between $\kappa$ and its ideal key comparison.
By Uhlmann's theorem~\cite{Uhl76}, applied to each diagonal key block,
there are normalized purifications $\ket{\xi_a}_{A'B'E_{\rm all}}$ of
$\kappa_{E_{\rm all}}$, with phases chosen so that
\[
 \frac1{\sqrt d}\sum_a\ip{\xi_a}{\psi_{aa}}
 =\sqrt{F\!\left(\kappa,
       \left(d^{-1}\sum_a\proj{a,a}\right)\otimes\kappa_{E_{\rm all}}\right)}
 \ge1-e.
\]
We may append an unused local pure ancilla and choose a common shield
space with a countably infinite-dimensional complement to each Schmidt
support. The partial isometries relating purifications of the same state
then extend to unitaries on this shield. Therefore
$\ket\Gamma=d^{-1/2}\sum_a\ket{a,a}\ket{\xi_a}$ reduces to a private state
$\gamma$ on $AA'BB'$. Fidelity monotonicity then gives
\[
 F(\omega,\gamma)\ge|\ip\Gamma\Psi|^2\ge(1-e)^2.
\]
Thus a key protocol of error $e$ supplies an approximate private state
with infidelity at most $2e-e^2$, without changing its optical inputs.
For entanglement generation the target already has the private-state
form, with trivial shields.

\begin{proposition}[Finite-block converse]
\label{prop:finite}
Consider an $n$-use adaptive protocol satisfying
\eqref{eq:energy-constraint}, whose final state has squared fidelity at
least $1-\varepsilon$ with a private state of key dimension $d\ge2$.
For $0\le\varepsilon\le\tfrac12$,
\begin{equation}
 (1-\varepsilon)\log_2d
 \le n f_\eta(N)+h_2(\varepsilon).
 \label{eq:finite-converse}
\end{equation}
\end{proposition}
\begin{proof}
Let $\nu_i=\Tr(\nhat\bar\rho_{S_i})$ be the unconditional mean energy at
use $i$. Keep all local memories, including the registers needed for
adaptivity. Corollary~\ref{cor:amortized} applies to the state before each
transmission, while the intermediate LOCC maps cannot increase $\ERi$.
The initial state has zero entanglement, so telescoping gives
\begin{equation}
 \ERi(\omega_{\mathrm{final}})\le\sum_{i=1}^n f_\eta(\nu_i).
 \label{eq:telescoping}
\end{equation}
All the quantities $\ER$ encountered here are finite: starting from
zero, the case $m=1$ of Lemma~\ref{lem:block} bounds their increase at
each use, and every $\nu_i$ is finite by~\eqref{eq:energy-constraint}.

The function $f_\eta$ is nondecreasing and concave on $[0,\infty)$.
For $x>0$,
\begin{equation}
 f_\eta''(x)=-\frac{\eta}{\ln2\,x(x+1)(qx+1)}\le0,
 \qquad \lim_{x\to\infty}f_\eta'(x)=0.
 \label{eq:concavity}
\end{equation}
The derivative is therefore nonnegative; continuity covers $x=0$.
Jensen's inequality and the energy budget give
\[
 \sum_{i=1}^n f_\eta(\nu_i)
 \le n f_\eta\!\left(\frac1n\sum_i\nu_i\right)
 \le n f_\eta(N).
\]
Apply Lemma~\ref{lem:privacy} to~\eqref{eq:telescoping} to obtain
\eqref{eq:finite-converse}.
\end{proof}

For operational key error $e$, Proposition~\ref{prop:finite} applies with
$\varepsilon=2e-e^2$ whenever this is at most $1/2$.
For any sequence of codes with $\varepsilon_n\to0$,
\eqref{eq:finite-converse} implies
$\limsup_n n^{-1}\log_2d_n\le f_\eta(N)$.
This proves the converse for both capacities in
Theorem~\ref{thm:main}.

\subsection{Achievability}
\label{sec:achievability}

Alice prepares independent copies of the two-mode squeezed vacuum
\begin{equation}
 \ket{\psi_N}_{RS}
 =\frac1{\sqrt{N+1}}\sum_{k\ge0}
   \left(\frac{N}{N+1}\right)^{k/2}\ket{k}_R\ket{k}_S
 \label{eq:tmsv}
\end{equation}
and sends $S$ through $\cL_\eta$. Each transmitted mode has mean photon
number $N$. For the output $\omega_{RB}$, its purification $RBE$ is
pure, $\omega_R=\tau_N$, and $\omega_E=\tau_{qN}$. Hence
\begin{equation}
 S(R)_\omega-S(RB)_\omega=g(N)-g(qN)=f_\eta(N).
 \label{eq:reverse-hashing}
\end{equation}
The hashing inequality, with classical communication from Bob to Alice,
achieves this rate~\cite{DW05,GPLS09}.

Here is a finite-dimensional justification that preserves the exact
energy budget. For $N>0$, truncate~\eqref{eq:tmsv} to $0\le k\le L$ and
normalize, obtaining $\ket{\psi_{N,L}}$. Its transmitted mean $N_L$ is
at most $N$ and tends to $N$. Put
$s_L=\Tr(\Pi_{\le L}\tau_N)$ and
$\tau_{N,L}=\Pi_{\le L}\tau_N\Pi_{\le L}/s_L$.
Then
\[
 D(\tau_{N,L}\Vert\tau_N)=-\log_2s_L\longrightarrow0,
 \qquad
 D(\cL_q(\tau_{N,L})\Vert\tau_{qN})\le-\log_2s_L.
\]
The logarithm of a positive-temperature thermal state is affine in
photon number. Since the two means tend to $N$ and $qN$, respectively,
these bounds imply
$S(\tau_{N,L})\to g(N)$ and
$S(\cL_q(\tau_{N,L}))\to g(qN)$; for $q=0$ the latter state is vacuum.
Finite-dimensional reverse hashing therefore achieves every rate below
$f_\eta(N)$. For $N=0$ the claimed rate is zero.
Finally, measuring the distilled maximally entangled states gives the
same number of secret bits. Thus
\[
 f_\eta(N)\le Q_2(\cL_\eta,N)\le K(\cL_\eta,N)\le f_\eta(N),
\]
which completes the proof of Theorem~\ref{thm:main}.

\paragraph{Note added and AI statement.} The main result of this work was obtained on 17 September 2026 with a single prompt of ChatGPT Astra (OpenAI). While refining the manuscript, we became aware of \cite{Pir26a}, whose author independently obtained the same result.

\end{document}